\documentclass[twocolumn,conference,letterpaper]{IEEEtran}
\usepackage[left=0.75in,right=0.75in,top=1in,bottom=0.75in]{geometry}

\usepackage{cite}

\usepackage{amssymb,amsmath,latexsym,amsfonts,amsthm,mathtools}

\DeclareMathOperator*{\argmin}{arg\,min}
\DeclareMathOperator{\im}{im}

\usepackage[colorlinks=true]{hyperref}
\hypersetup{colorlinks, breaklinks, citecolor=blue, linkcolor=blue, urlcolor=blue}
\usepackage{graphicx}
\graphicspath{{figures/}{results_data_informative_safety/}}
\usepackage{float,subcaption}
\usepackage{textcomp}
\usepackage{xcolor}
\definecolor{darkpastelpurple}{rgb}{0.59, 0.44, 0.84}
\usepackage[nameinlink]{cleveref}
\Crefformat{figure}{#2Fig.~#1#3}
\Crefmultiformat{figure}{Figs.~#2#1#3}{ and~#2#1#3}{, #2#1#3}{ and~#2#1#3}

\usepackage{tikz}
\usetikzlibrary{automata, shapes, arrows, calc, arrows.meta, fit, positioning}

\theoremstyle{plain}
\newtheorem{theorem}{Theorem}
\newtheorem{lemma}{Lemma}
\newtheorem{corollary}{Corollary}

\newtheorem*{problem*}{Problem}
\newtheorem{problem}{Problem}
\Crefname{problem}{Problem}{Problems} 
\crefname{problem}{problem}{problems}     

\theoremstyle{remark}
\newtheorem{remark}{Remark}
\newtheorem{assumption}{Assumption}
\Crefname{assumption}{Assumption}{Assumptions} 
\crefname{assumption}{assumption}{assumptions}     

\newtheorem{definition}{Definition}
\theoremstyle{definition}

\usepackage{mathrsfs}

\usepackage{newtxmath}
\usepackage{booktabs}

\IEEEoverridecommandlockouts

\begin{document}
\title{Finite-Data Safety Informativity Under Dynamic Asymmetric Actuation}
\author{Abhinav Sinha,~\IEEEmembership{Senior Member,~IEEE}, Praveen Kumar Ranjan,~\IEEEmembership{Member,~IEEE},\\
and Yongcan Cao,~\IEEEmembership{Senior Member,~IEEE}%
\thanks{A. Sinha is with the GALACxIS Lab, Department of Aerospace Engineering, University of Cincinnati, Cincinnati, OH 45221, USA (e-mail: abhinav.sinha@uc.edu). P. K. Ranjan and Y. Cao are with the Unmanned Systems Lab (USL), Department of Electrical Engineering, The University of Texas at San Antonio, TX 78249, USA (email:praveen.ranjan@my.utsa.edu, yongcan.cao@utsa.edu).}
}
\maketitle
\thispagestyle{empty}

\begin{abstract}
When the system model is not fully known, measurement error and limited excitation can leave several models consistent with the same finite data. A command judged safe for one model may fail for another, while limited control authority can prevent the corrective action needed to preserve safety. To ensure safety under model uncertainty and asymmetric input limits, we develop a finite-data certificate that determines whether a command can enforce a prescribed safety inequality. For a linearly parameterized safety channel with exactly known regressors and bounded aggregate residual error, we derive a support formula for the worst-case safety contribution of all data-consistent models. The formula identifies the regressor directions that admit a finite bound, allowing rank-deficient records to contribute to safety certification. Using certified componentwise bounds on actuator tracking error yields an affine inequality with a necessary and sufficient test for pointwise command feasibility. The affine inequality reduces computation of the closest certified command to a scalar root-finding problem. It also yields a closed-form gate that selects the largest certified fraction of a prescribed command segment. The proposed certificate guarantees output safety within its operating domain, provided the feedback is locally Lipschitz and the uncertainty bounds remain valid. Domain retention and full-state continuation extend this guarantee to all time. A vehicle study demonstrates that output safety can be certified from finite measurements in a safety-critical setting with model and actuator uncertainty.
\end{abstract}
\setlength{\abovedisplayskip}{4pt plus 1pt minus 1pt}
\setlength{\belowdisplayskip}{4pt plus 1pt minus 1pt}
\setlength{\belowdisplayshortskip}{4pt plus 1pt minus 1pt}
\setlength{\textfloatsep}{10pt plus 2pt minus 2pt}
\setlength{\floatsep}{6pt plus 1pt minus 1pt}
\captionsetup[figure]{skip=4pt}

    \section{Introduction}
In many applications, a controller must meet a performance objective while respecting constraints that define safe operation. A vehicle, for example, may need to follow a reference trajectory while remaining inside a moving corridor. Near the corridor boundary, the controller must retain enough authority to counter disturbances without violating the input limits. When control decisions rely on finite measurements, the data must support the safety of a proposed correction, and the actuators must be able to realize it.

Set-invariance theory and control barrier functions provide conditions for selecting corrective inputs that keep trajectories within admissible sets \cite{Blanchini1999Survey,Ames2017CBF}. Safety filters use these conditions to modify a nominal command through an instantaneous inequality \cite{Ames2017CBF} or predictive optimization \cite{Wabersich2021PSF}. Higher-order barrier designs address constraints of high relative degree \cite{Xiao2022HOCBF} and the feasibility of corrective inputs under tight input bounds \cite{Xiao2022Feasibility}. These safety guarantees depend on a model or uncertainty bounds that remain valid during operation. When a reliable model is unavailable, unknown contributions from the dynamics prevent direct evaluation of these conditions.

Measurements can be used to estimate the unknown terms and quantify the uncertainty relevant to safety certification. Gaussian-process methods use prediction bounds to account for uncertainty in barrier conditions \cite{Jagtap2020GPCBF,Dhiman2023BayesianCBF}, while reinforcement learning can estimate corrections to the model terms entering safety constraints \cite{Choi2020RLCBFCLF}. Conformal calibration connects learned models to probabilistic guarantees \cite{HsuTsukamoto2025Conformal}. Robust adaptive and uncertainty-estimation methods combine estimates with certified error bounds \cite{Lopez2021RobustAdaptiveCBF,Das2025UncertaintyCBF}. However, accurate prediction and uncertainty quantification do not establish whether a given finite record is sufficient for safety certification under limited control authority. A deterministic certificate based on a finite record must establish that an admissible command satisfies the safety inequality for every model consistent with the measurements and their error bound.

Data informativity formalizes this question of sufficiency by assessing whether a record supports a specified task for all data-consistent systems \cite{VanWaarde2020Informativity,VanWaarde2023CSM}. In linear feedback synthesis, data parameterizations and matrix inequalities yield conditions for such sufficiency under bounded errors \cite{DePersis2020Formulas,VanWaarde2022MatrixSLemma,Bisoffi2022Petersen}. For nonlinear systems, set-membership and barrier-learning methods use structural and regularity assumptions to obtain deterministic guarantees from data \cite{Martin2023Survey,Strong2025DeterministicCBF}. A further challenge is to preserve the safety guarantee when the command is realized by a dynamic actuator. Tracking error can invalidate the safety inequality, while asymmetric limits restrict the corrective action available to maintain safety.

Actuator-aware barrier formulations address this challenge through optimization-based filters that account for dynamically generated inputs or unmodeled input dynamics \cite{Ames2020IntegralCBF,Seiler2021IQC}. Admissibility-preserving input realization (APIR) embeds asymmetric limits directly in a smooth actuator map, keeping the realized input admissible by construction \cite{Kumar2026APC,SinhaKumar2026AsymmetricConsensus,Sinha2026DirectedSafeCoordination}. Within this realization, certified tracking-error bounds quantify a margin in the safety inequality that accounts for the difference between commanded and realized inputs.

Output-safety certification from finite data still faces an obstacle in the sense that the corrective action required by the data-consistent models may exceed the authority available under actuator tracking error and asymmetric input limits. Addressing this problem requires bounding the safety contribution of the data-consistent models without full parameter identification and comparing that bound with the authority available under actuator tracking error. In such scenarios, a feasible certificate must be converted into command-selection rules with the regularity needed for trajectory safety. We address these requirements through the following contributions.
\begin{itemize}
    \item We characterize the minimum residual-error budget for data consistency and quantify the amplification of uncertainty at each queried regressor. For a consistent record and a nonzero safety functional, the result identifies the unmeasured regressor directions that make the worst-case model contribution unbounded.
    \item We quantify the loss of certified control authority due to actuator tracking error as an additive reserve. Subtracting this reserve from the static finite-data authority characterizes pointwise feasibility for the chosen certificate and recovers the static certificate when the error bounds vanish.
    \item For implementation, we derive the command projection and admission gate with sufficient conditions for local Lipschitz continuity. A counterexample establishes that pointwise certificate feasibility can hold even when no continuous feedback selector exists.
    \item We validate the certificate numerically on a vehicle model by comparing coverage across equal-sized measurement records and against static and exact-parameter benchmarks. Closed-loop checks assess output and input constraint satisfaction, and halving the integration step preserves sampled certificate feasibility.
\end{itemize}

\section{Problem Formulation}
Our objective is to decide whether finite measurements certify an admissible command at a given state and time. This command must enforce a prescribed output-safety inequality for every system consistent with those measurements and every actuator error within a known bound. At feasible queries, we seek the closest certified command to a baseline and the largest admissible fraction of a desired command. The formulation below expresses output safety through a scalar storage function and identifies the uncertainties that enter its derivative.

The notation $\|\cdot\|$ denotes the Euclidean vector norm and $\|\cdot\|_{\mathrm{F}}$ the Frobenius norm. For a matrix $\mathbf{A}$, $\mathbf{A}^{\dagger}$ is the Moore--Penrose pseudoinverse, $\im(\mathbf{A})$ its column space, and $\ker(\mathbf{A})$ its nullspace. Identity matrices have compatible dimensions; $\mathbf{A}\succ\mathbf{0}$ denotes symmetric positive definiteness.

We describe the system by a nominal model plus a mismatch. Let $\mathbf{x}\in\mathbb{R}^{n}$ and $\mathbf{u}\in\mathbb{R}^{m}$ denote the state and realized input, where $n,m\geq1$, with time $t\geq0$ and initial certification time $t_{0}\geq0$. The true drift $\mathbf{f}^{\star}$ and input matrix $\mathbf{G}^{\star}$ have nominal counterparts $\mathbf{f}_{0}\in\mathbb{R}^{n}$ and $\mathbf{G}_{0}\in\mathbb{R}^{n\times m}$. The mismatch $\boldsymbol{\Delta}:=\mathbf{f}^{\star}-\mathbf{f}_{0}+(\mathbf{G}^{\star}-\mathbf{G}_{0})\mathbf{u}+\mathbf{w}$ collects the modeling errors and disturbance $\mathbf{w}\in\mathbb{R}^{n}$. The system satisfies
\begin{equation}
\dot{\mathbf{x}}=\mathbf{f}_{0}(\mathbf{x},t)
+\mathbf{G}_{0}(\mathbf{x},t)\mathbf{u}
+\boldsymbol{\Delta}(\mathbf{x},\mathbf{u},t).
\label{eq:plant}
\end{equation}
\begin{assumption}
\label{ass:plant-regularity}
The vector field in \eqref{eq:plant} is measurable in time, locally bounded, and locally Lipschitz in state and input uniformly on compact time intervals. For all-time claims, we additionally assume that bounded inputs cannot cause finite-time escape of the full state. A locally integrable linear-growth bound on the system vector field is sufficient for this continuation property.
\end{assumption}
Under locally Lipschitz feedback and fixed admissible disturbances, the regularity conditions in \Cref{ass:plant-regularity} ensure a unique closed-loop solution on an initial time interval. Output constraints may leave some state components unrestricted. The continuation condition rules out finite-time escape of the full state under bounded inputs.

To enforce the output limits, we use logarithmic coordinates that send the lower and upper boundaries to $-\infty$ and $+\infty$, respectively \cite{SinhaKumar2026AsymmetricConsensus,Kumar2026APC,Sinha2026DirectedSafeCoordination}. Keeping these coordinates bounded will keep the output strictly inside its permitted interval.
\begin{definition}
\label{def:safety-coordinates}
For $q\geq1$, we let the output map $\mathbf{h}:\mathbb{R}^{n}\times\mathbb{R}_{\geq0}\to\mathbb{R}^{q}$ and bounds $\underline{y}_{\ell},\overline{y}_{\ell}$, $\ell=1,\ldots,q$, be $C^{1}$, with $\overline{y}_{\ell}(t)-\underline{y}_{\ell}(t)\geq\underline{\varrho}_{\ell}>0$ for constants $\underline{\varrho}_{\ell}$. Let us define the normalized outputs $s_{\ell}:=(h_{\ell}-\underline{y}_{\ell})/(\overline{y}_{\ell}-\underline{y}_{\ell})$. On the safe set $\mathcal{C}(t):=\{\mathbf{x}:0<s_{\ell}(\mathbf{x},t)<1,~\forall\ell\}$, the logit coordinates are $z_{\ell}:=\log(s_{\ell}/(1-s_{\ell}))$ and $\mathbf{z}:=[z_{1},\ldots,z_{q}]^{\top}$. For a uniformly bounded $C^{1}$ reference $\mathbf{z}_{r}(t)\in\mathbb{R}^{q}$ and constant $\mathbf{P}\in\mathbb{R}^{q\times q}$ with $\mathbf{P}\succ\mathbf{0}$, we set $\mathbf{e}_{z}:=\mathbf{z}-\mathbf{z}_{r}$, $W:=\mathbf{e}_{z}^{\top}\mathbf{P}\mathbf{e}_{z}/2$, and $\mathbf{c}:=\mathbf{P}\mathbf{e}_{z}$.
\end{definition}
Bounded $W$ keeps the normalized outputs away from zero and one, providing the link between dissipation and strict output safety. We write $\mathbf{T}(\mathbf{x},t):=\mathbf{z}(\mathbf{x},t)$, with state Jacobian $\mathbf{T}_{\mathbf{x}}$ and partial time derivative $\mathbf{T}_{t}$. The nominal transformed drift, input matrix, and projected mismatch are $\mathbf{a}:=\mathbf{T}_{\mathbf{x}}\mathbf{f}_{0}+\mathbf{T}_{t}-\dot{\mathbf{z}}_{r}\in\mathbb{R}^{q}$, $\mathbf{B}:=\mathbf{T}_{\mathbf{x}}\mathbf{G}_{0}\in\mathbb{R}^{q\times m}$, and $\boldsymbol{\delta}:=\mathbf{T}_{\mathbf{x}}\boldsymbol{\Delta}\in\mathbb{R}^{q}$. Differentiating the coordinates in \Cref{def:safety-coordinates} along \eqref{eq:plant} yields
\begin{equation}
\dot{\mathbf{e}}_{z}=\mathbf{a}(\mathbf{x},t)
+\mathbf{B}(\mathbf{x},t)\mathbf{u}
+\boldsymbol{\delta}(\mathbf{x},\mathbf{u},t).
\label{eq:channel}
\end{equation}
We use measurements of the transformed error derivative $\dot{\mathbf{e}}_{z}$, recorded with the corresponding state $\mathbf{x}$ and realized input $\mathbf{u}$, to form residual observations of the mismatch in \eqref{eq:channel}. We represent this mismatch using known feature functions and unknown constant coefficients, with a bounded remainder accounting for terms outside that representation.
\begin{assumption}
\label{ass:parameterization}
For a regressor dimension $p\geq1$, a known regressor $\boldsymbol{\phi}(\mathbf{x},t)\in\mathbb{R}^{p}$, one unknown constant matrix $\boldsymbol{\Theta}^{\star}\in\mathbb{R}^{p\times q}$, and a remainder $\mathbf{d}\in\mathbb{R}^{q}$ with known bound $\bar{d}\geq0$, the identity $\boldsymbol{\delta}(\mathbf{x},\mathbf{u},t)=\boldsymbol{\Theta}^{\star\top}\boldsymbol{\phi}(\mathbf{x},t)+\mathbf{d}(\mathbf{x},\mathbf{u},t)$ holds throughout the operating domain and for every admissible realized input, with $\|\mathbf{d}\|\leq\bar{d}$. The same parameter matrix applies to all samples and subsequent queries, and the regressors are known without error.
\end{assumption}
\Cref{ass:parameterization} applies when the feature functions are known and their coefficients remain fixed during data collection and subsequent operation. For example, the features may describe wind patterns with unknown constant amplitudes. The remainder bound must hold in the transformed coordinates throughout the certification domain. Since the logit Jacobian grows near an output boundary, a bounded physical disturbance need not produce a uniform $\bar{d}$ on the entire safe set. The choice of operating domain is part of justifying this bound.

The regressor in \Cref{ass:parameterization} is independent of the input, so all input dependence of the mismatch is contained in $\mathbf{d}$. At fixed $(\mathbf{x},t)$, bounding $\mathbf{c}^{\top}\mathbf{d}$ by $\bar{d}\|\mathbf{c}\|$ leaves $\mathbf{c}^{\top}\mathbf{B}\mathbf{u}$ as the only input-dependent term in the upper bound on $\dot{W}$. This structure makes the eventual command constraint affine.

We next distinguish the requested input from the input delivered by the actuator. A smooth saturation map enforces the physical limits. A bound on its internal tracking error will account for the resulting input mismatch in the safety inequality.
\begin{assumption}
\label{ass:actuator}
For $j=1,\ldots,m$, we take the limits $\underline{u}_{j}<\overline{u}_{j}$ to be finite and the map $\sigma_{j}:\mathbb{R}\to(\underline{u}_{j},\overline{u}_{j})$ to be onto and $C^{1}$, with $0<\sigma'_{j}(\nu)\leq L_{\sigma,j}$ for known $L_{\sigma,j}>0$. We write $\chi_{j}:=\sigma_{j}^{-1}$ and use $\boldsymbol{\sigma}$ and $\boldsymbol{\chi}$ componentwise. For $0<\varepsilon_{j}<(\overline{u}_{j}-\underline{u}_{j})/2$, we define the inner command set $\mathcal{U}_{\varepsilon}:=\prod_{j=1}^{m}[\underline{u}_{j}+\varepsilon_{j},\overline{u}_{j}-\varepsilon_{j}]$. Given a command $\mathbf{u}_{c}\in\mathcal{U}_{\varepsilon}$ and internal state $\boldsymbol{\xi}\in\mathbb{R}^{m}$, we set $\boldsymbol{\nu}_{c}:=\boldsymbol{\chi}(\mathbf{u}_{c})\in\mathbb{R}^{m}$ and tracking error $\mathbf{e}_{a}:=\boldsymbol{\xi}-\boldsymbol{\nu}_{c}$. Then, the realized input becomes
\begin{equation}
\mathbf{u}=\boldsymbol{\sigma}(\boldsymbol{\nu}_{c}+\mathbf{e}_{a}).
\label{eq:actuator}
\end{equation}
With known rates $a_{j}>0$, we write $\boldsymbol{\Lambda}_{a}:=\operatorname{diag}(a_{1},\ldots,a_{m})$. A measurable effective disturbance $\boldsymbol{\omega}_{a}\in\mathbb{R}^{m}$ with known bounds $\bar{\omega}_{a,j}\geq0$ drives the error dynamics
\begin{equation}
\dot{\mathbf{e}}_{a}=-\boldsymbol{\Lambda}_{a}\mathbf{e}_{a}
+\boldsymbol{\omega}_{a};
~~|\omega_{a,j}(t)|\leq\bar{\omega}_{a,j}.
\label{eq:ea}
\end{equation}
The initial error $\mathbf{e}_{a}(t_{0})$ is finite and known.
\end{assumption}
For an inner-loop disturbance $\mathbf{w}_{a}(t)\in\mathbb{R}^{m}$, a conventional lag $\dot{\boldsymbol{\xi}}=-\boldsymbol{\Lambda}_{a}(\boldsymbol{\xi}-\boldsymbol{\nu}_{c})+\mathbf{w}_{a}$ has $\boldsymbol{\omega}_{a}=\mathbf{w}_{a}-\dot{\boldsymbol{\nu}}_{c}$ almost everywhere. The bound in \Cref{ass:actuator} must cover the commanded motion as well as the inner-loop disturbance. Then, it follows that a certified command-rate bound or derivative feedforward with a bounded residual can justify this requirement. Restricting $\mathbf{u}_{c}$ to the inner command set limits the magnitude of the internal command $\boldsymbol{\nu}_{c}=\boldsymbol{\chi}(\mathbf{u}_{c})$, but provides no bound on its time derivative $\dot{\boldsymbol{\nu}}_{c}$.

When the limits straddle zero, setting $L_{j}:=-\underline{u}_{j}>0$ and $U_{j}:=\overline{u}_{j}>0$ allows an admissible map to be written as
\begin{align}
    \sigma_{j}(\nu)=&-L_{j}+\dfrac{U_{j}+L_{j}}{1+(U_{j}/L_{j})\exp(-(U_{j}+L_{j})\nu/(U_{j}L_{j}))},
    \label{eq:asymmetric-map}
\end{align}
with $L_{\sigma,j}=(U_{j}+L_{j})^{2}/(4U_{j}L_{j})$.
The next lemma converts the internal tracking-error bound into an interval for the physical-input error. The decaying term accounts for the initial actuator error, whereas the remaining term accounts for the disturbance accumulated during operation.
\begin{lemma}
\label{lem:actuator-tube}
Under \Cref{ass:actuator}, we set $\tau:=t-t_{0}\geq0$ and define $\bar{e}_{a,j}(t):=\exp(-a_{j}\tau)|e_{a,j}(t_{0})|+(\bar{\omega}_{a,j}/a_{j})(1-\exp(-a_{j}\tau))$ and $\varepsilon_{a,j}(t):=L_{\sigma,j}\bar{e}_{a,j}(t)$. Then
\begin{equation}
|e_{a,j}(t)|\leq\bar{e}_{a,j}(t),
~~|u_{j}(t)-u_{c,j}(t)|\leq\varepsilon_{a,j}(t).
\label{eq:tube}
\end{equation}
The input satisfies $\underline{u}_{j}<u_{j}(t)<\overline{u}_{j}$ for all finite actuator states.
\end{lemma}
\begin{proof}
Variation of constants in \eqref{eq:ea}, with the disturbance bound from \Cref{ass:actuator}, bounds $|e_{a,j}(t)|$ by $\bar{e}_{a,j}(t)$. Applying the Lipschitz bound in \Cref{ass:actuator} to \eqref{eq:actuator} yields $|u_{j}-u_{c,j}|\leq L_{\sigma,j}|e_{a,j}|$, proving \eqref{eq:tube}. The range of $\sigma_{j}$ in \Cref{ass:actuator} ensures strict amplitude admissibility.
\end{proof}
For the physical-input error $\mathbf{e}_{u}:=\mathbf{u}-\mathbf{u}_{c}$, \eqref{eq:tube} ensures that $\mathbf{e}_{u}(t)$ belongs to $\mathcal{E}_{a}(t):=\{\mathbf{e}_{u}\in\mathbb{R}^{m}:|e_{u,j}|\leq\varepsilon_{a,j}(t),~j=1,\ldots,m\}$. The safety certificate maximizes the actuator-error contribution $\mathbf{c}^{\top}\mathbf{B}\mathbf{e}_{u}$ over $\mathcal{E}_{a}(t)$. Since the actuator-error set may contain errors unattainable under \eqref{eq:actuator}, the resulting bound can be conservative.

The data record must describe the input actually received by the system. We use the realized input to form each residual, so actuator tracking error is accounted for through the recorded motion. For $N\geq1$ samples, indexed by $k=1,\ldots,N$, we use $t_{k}$ for the sampling time and set $\mathbf{x}_{k}:=\mathbf{x}(t_{k})\in\mathcal{C}(t_{k})$, $\mathbf{u}_{k}:=\mathbf{u}(t_{k})$, and $\boldsymbol{\phi}_{k}:=\boldsymbol{\phi}(\mathbf{x}_{k},t_{k})$. A measured derivative $\hat{\dot{\mathbf{e}}}_{z,k}\in\mathbb{R}^{q}$ determines the residual $\mathbf{r}_{k}:=\hat{\dot{\mathbf{e}}}_{z,k}-\mathbf{a}(\mathbf{x}_{k},t_{k})-\mathbf{B}(\mathbf{x}_{k},t_{k})\mathbf{u}_{k}$. We collect the sampled regressors in $\boldsymbol{\Phi}_{N}:=[\boldsymbol{\phi}_{1},\ldots,\boldsymbol{\phi}_{N}]\in\mathbb{R}^{p\times N}$ and the residuals in $\mathbf{R}_{N}:=[\mathbf{r}_{1},\ldots,\mathbf{r}_{N}]\in\mathbb{R}^{q\times N}$.
\begin{assumption}
\label{ass:data}
The aggregate remainder and measurement-error matrix $\mathbf{D}_{N}\in\mathbb{R}^{q\times N}$ satisfies $\mathbf{R}_{N}=\boldsymbol{\Theta}^{\star\top}\boldsymbol{\Phi}_{N}+\mathbf{D}_{N}$ and $\|\mathbf{D}_{N}\|_{\mathrm{F}}\leq\eta_{N}$ for a known $\eta_{N}\geq0$.
\end{assumption}
For example, if the complete residual-measurement error matrix has Frobenius norm at most $\bar{n}_{N}\geq0$, then one may take $\eta_{N}=\sqrt{N}\bar{d}+\bar{n}_{N}$. Errors in the regressors require an additional uncertainty description; with unbounded parameters, their effects generally cannot be absorbed into a known additive radius. The data-consistent parameters form the set
\begin{equation}
\mathcal{T}_{N}:=\left\{\boldsymbol{\Theta}\in\mathbb{R}^{p\times q}:
\|\mathbf{R}_{N}-\boldsymbol{\Theta}^{\top}\boldsymbol{\Phi}_{N}\|_{\mathrm{F}}
\leq\eta_{N}\right\}.
\label{eq:TN}
\end{equation}
\Cref{ass:data} ensures $\boldsymbol{\Theta}^{\star}\in\mathcal{T}_{N}$ without a rank requirement on $\boldsymbol{\Phi}_{N}$.

We seek a command that enforces $\dot{W}+\alpha(W)\leq0$, where $\alpha:\mathbb{R}_{\geq0}\to\mathbb{R}_{\geq0}$ is continuous and strictly increasing, with $\alpha(0)=0$. This inequality keeps $W$ bounded and hence preserves the output limits. A certified command must satisfy the inequality for every uncertainty allowed by the data and actuator bounds. Suppressing common state and time arguments, we express this requirement through the worst-case dissipation functional
\begin{multline}
\mathcal{M}_{N}^{a}(\mathbf{x},\mathbf{u}_{c},t):=\\
\sup_{\substack{\boldsymbol{\Theta}\in\mathcal{T}_{N},~\|\mathbf{d}\|\leq\bar{d}\\
\mathbf{e}_{u}\in\mathcal{E}_{a}(t)}}
\left\{\mathbf{c}^{\top}\bigl[\mathbf{a}+\mathbf{B}(\mathbf{u}_{c}+\mathbf{e}_{u})
+\boldsymbol{\Theta}^{\top}\boldsymbol{\phi}+\mathbf{d}\bigr]+\alpha(W)\right\}.
\label{eq:robust-dissipation}
\end{multline}
Following the informativity framework \cite{VanWaarde2020Informativity,VanWaarde2023CSM}, we require a common feedback law to satisfy the chosen certificate for all data-consistent models.
\begin{definition}
\label{def:safety-informativity}
The certified command set is $\mathcal{K}_{N}^{a}(\mathbf{x},t):=\{\mathbf{u}_{c}\in\mathcal{U}_{\varepsilon}:\mathcal{M}_{N}^{a}(\mathbf{x},\mathbf{u}_{c},t)\leq0\}$. We let $\Omega(t)\subseteq\mathcal{C}(t)$ define a relatively open operating domain $\mathcal{D}_{\Omega}:=\{(\mathbf{x},t):t\geq t_{0},~\mathbf{x}\in\Omega(t)\}$. The data are informative for the chosen certificate on $\mathcal{D}_{\Omega}$ if a locally Lipschitz law $\boldsymbol{\kappa}_{N}:\mathcal{D}_{\Omega}\to\mathcal{U}_{\varepsilon}$ satisfies $\boldsymbol{\kappa}_{N}(\mathbf{x},t)\in\mathcal{K}_{N}^{a}(\mathbf{x},t)$ throughout $\mathcal{D}_{\Omega}$.
\end{definition}
\begin{definition}
\label{def:pointwise-informativity}
The data and actuator tube are pointwise informative for the chosen certificate at $(\mathbf{x},t)$ if $\mathcal{K}_{N}^{a}(\mathbf{x},t)\neq\varnothing$.
\end{definition}
Pointwise informativity ensures that a certified command exists at a given state and time. Certificate informativity requires a locally Lipschitz feedback law that selects certified commands throughout the operating domain, providing the feedback regularity used in the closed-loop existence and uniqueness analysis.
\begin{problem}
\label{prob:safety-synthesis}
Under \Cref{ass:plant-regularity,ass:parameterization,ass:actuator,ass:data}, given the data $(\boldsymbol{\Phi}_{N},\mathbf{R}_{N},\eta_{N})$ and actuator-error set $\mathcal{E}_{a}(t)$, we seek to characterize the queries $(\mathbf{x},t)\in\mathcal{D}_{\Omega}$ satisfying \Cref{def:pointwise-informativity}. Let $\mathbf{u}_{0}(\mathbf{x},t)\in\mathbb{R}^{m}$ be a baseline command and $\mathbf{u}_{d}(\mathbf{x},t)\in\mathcal{U}_{\varepsilon}$ a desired learned command. At each feasible query, we seek $\mathbf{u}_{s}\in\mathcal{K}_{N}^{a}$ minimizing $\|\mathbf{u}_{s}-\mathbf{u}_{0}\|^{2}/2$ and the largest $\lambda\in[0,1]$ for which $\mathbf{u}_{c,\lambda}:=(1-\lambda)\mathbf{u}_{s}+\lambda\mathbf{u}_{d}\in\mathcal{K}_{N}^{a}$.
\end{problem}
The trajectory guarantee will apply while the state remains in $\mathcal{D}_{\Omega}$. An all-time conclusion also requires the full-state continuation property in \Cref{ass:plant-regularity}.

\section{Main Results}
Each row of the modeled residual matrix $\boldsymbol{\Theta}^{\top}\boldsymbol{\Phi}_{N}$ lies in the row space of $\boldsymbol{\Phi}_{N}$. Let $\mathbf{P}_{N}:=\boldsymbol{\Phi}_{N}^{\dagger}\boldsymbol{\Phi}_{N}\in\mathbb{R}^{N\times N}$ be the orthogonal projector onto this space. The orthogonal component of the measured residuals cannot be fitted by changing $\boldsymbol{\Theta}$, so we define its unavoidable norm as $\eta_{N}^{\min}:=\|\mathbf{R}_{N}(\mathbf{I}-\mathbf{P}_{N})\|_{\mathrm{F}}$. When $\eta_{N}\geq\eta_{N}^{\min}$, the remaining Frobenius-norm budget is $\chi_{N}:=\sqrt{\eta_{N}^{2}-(\eta_{N}^{\min})^{2}}$.
\begin{theorem}
\label{thm:data-support}
The set \eqref{eq:TN} is nonempty if and only if $\eta_{N}\geq\eta_{N}^{\min}$. Suppose $\mathcal{T}_{N}\neq\varnothing$. At a query with $\boldsymbol{\phi}\in\im(\boldsymbol{\Phi}_{N})$, define $\boldsymbol{\beta}_{N}:=\boldsymbol{\Phi}_{N}^{\dagger}\boldsymbol{\phi}\in\mathbb{R}^{N}$. Then the worst-case model contribution is
\begin{equation}
\sup_{\boldsymbol{\Theta}\in\mathcal{T}_{N}}
\mathbf{c}^{\top}\boldsymbol{\Theta}^{\top}\boldsymbol{\phi}
=\mathbf{c}^{\top}\mathbf{R}_{N}\boldsymbol{\beta}_{N}
+\chi_{N}\|\mathbf{c}\|\|\boldsymbol{\beta}_{N}\|.
\label{eq:support}
\end{equation}
For $\mathbf{c}=\mathbf{0}$, the support is zero for every $\boldsymbol{\phi}$. For $\mathbf{c}\neq\mathbf{0}$ and $\boldsymbol{\phi}\notin\im(\boldsymbol{\Phi}_{N})$, the support equals $+\infty$.
\end{theorem}
\begin{proof}
We separate the part of the residual that a parameter matrix can fit from the part fixed by the data. For any $\boldsymbol{\Theta}\in\mathbb{R}^{p\times q}$, we write the row-space residual as $\mathbf{H}:=\mathbf{R}_{N}\mathbf{P}_{N}-\boldsymbol{\Theta}^{\top}\boldsymbol{\Phi}_{N}$. Since $\mathbf{H}=\mathbf{H}\mathbf{P}_{N}$, the row-space residual is orthogonal to $\mathbf{R}_{N}(\mathbf{I}-\mathbf{P}_{N})$, and
\begin{equation}
\|\mathbf{R}_{N}-\boldsymbol{\Theta}^{\top}\boldsymbol{\Phi}_{N}\|_{\mathrm{F}}^{2}
=(\eta_{N}^{\min})^{2}+\|\mathbf{H}\|_{\mathrm{F}}^{2}.
\label{eq:residual-decomposition}
\end{equation}
For every row-space matrix $\mathbf{H}$, one may obtain a realizing parameter matrix from $\boldsymbol{\Theta}^{\top}=(\mathbf{R}_{N}\mathbf{P}_{N}-\mathbf{H})\boldsymbol{\Phi}_{N}^{\dagger}$. It follows from \eqref{eq:residual-decomposition} that the data are inconsistent when $\eta_{N}<\eta_{N}^{\min}$. Otherwise, \eqref{eq:TN} reduces to $\|\mathbf{H}\|_{\mathrm{F}}\leq\chi_{N}$, and $\mathbf{H}=\mathbf{0}$ attains the consistency threshold.

When $\boldsymbol{\phi}\in\im(\boldsymbol{\Phi}_{N})$, the identities $\boldsymbol{\Phi}_{N}\boldsymbol{\beta}_{N}=\boldsymbol{\phi}$ and $\mathbf{P}_{N}\boldsymbol{\beta}_{N}=\boldsymbol{\beta}_{N}$ reduce the objective to $\mathbf{c}^{\top}\mathbf{R}_{N}\boldsymbol{\beta}_{N}-\mathbf{c}^{\top}\mathbf{H}\boldsymbol{\beta}_{N}$. The Frobenius Cauchy--Schwarz inequality and \eqref{eq:residual-decomposition} bound the second term by $\chi_{N}\|\mathbf{c}\|\|\boldsymbol{\beta}_{N}\|$. If both vector norms are nonzero, the feasible choice $\mathbf{H}^{\star}:=-\chi_{N}\mathbf{c}\boldsymbol{\beta}_{N}^{\top}/(\|\mathbf{c}\|\|\boldsymbol{\beta}_{N}\|)$ attains the bound in \eqref{eq:support}. If either norm vanishes, the functional is zero.

For an out-of-span query, we set $\mathbf{v}:=(\mathbf{I}-\boldsymbol{\Phi}_{N}\boldsymbol{\Phi}_{N}^{\dagger})\boldsymbol{\phi}\neq\mathbf{0}$. Given $\boldsymbol{\Theta}_{0}\in\mathcal{T}_{N}$ and $s\geq0$, the perturbation $\boldsymbol{\Theta}_{s}:=\boldsymbol{\Theta}_{0}+s\mathbf{v}\mathbf{c}^{\top}$ preserves \eqref{eq:TN} because $\mathbf{v}^{\top}\boldsymbol{\Phi}_{N}=\mathbf{0}$. For $\mathbf{c}\neq\mathbf{0}$, the queried value grows by $s\|\mathbf{c}\|^{2}\|\mathbf{v}\|^{2}$ and tends to $+\infty$ as $s\to+\infty$, establishing that the support equals $+\infty$.
\end{proof}
The vector $\boldsymbol{\beta}_{N}$ expresses the queried feature using the measured features, whose norm determines the amplification of residual uncertainty in \eqref{eq:support}. Thus, equal sample counts and full row rank can coexist with very different uncertainty penalties.

We can now compare the uncertainty penalty with the available input authority. Let $\mathbf{b}_{N}:=\mathbf{B}^{\top}\mathbf{c}\in\mathbb{R}^{m}$ denote the input direction in the storage derivative. At an in-span query, collect the command-independent terms as
\begin{equation*}
d_{N}:=\mathbf{c}^{\top}\mathbf{a}+\mathbf{c}^{\top}\mathbf{R}_{N}\boldsymbol{\beta}_{N}+\alpha(W)+\chi_{N}\|\mathbf{c}\|\|\boldsymbol{\beta}_{N}\|+\bar{d}\|\mathbf{c}\|.
\end{equation*}
Maximizing over the actuator-error set yields the actuator reserve
\begin{equation*}
\rho_{a,N}:=\sum_{j=1}^{m}|b_{N,j}|\varepsilon_{a,j}.
\end{equation*}
Substituting this reserve and \eqref{eq:support} into \eqref{eq:robust-dissipation} leaves an affine function of the command:
\begin{equation}
\mathcal{M}_{N}^{a}(\mathbf{x},\mathbf{u}_{c},t)
=\mathfrak{m}_{N}^{a}(\mathbf{u}_{c})
:=d_{N}+\mathbf{b}_{N}^{\top}\mathbf{u}_{c}+\rho_{a,N}.
\label{eq:margin}
\end{equation}
At $\mathbf{c}=\mathbf{0}$, we set $d_{N}=0$ even for an out-of-span query; then $W=0$ and the margin is zero. For $\mathbf{c}\neq\mathbf{0}$ and an out-of-span query, \Cref{thm:data-support} implies $\mathcal{M}_{N}^{a}=+\infty$ and $\mathcal{K}_{N}^{a}=\varnothing$.

Feasibility is decided by minimizing \eqref{eq:margin} over the command set. For a finite margin, we choose the extreme command $\mathbf{u}_{N}^{\mathrm{ext}}\in\mathcal{U}_{\varepsilon}$ componentwise: $u_{N,j}^{\mathrm{ext}}:=\underline{u}_{j}+\varepsilon_{j}$ for $b_{N,j}\geq0$ and $u_{N,j}^{\mathrm{ext}}:=\overline{u}_{j}-\varepsilon_{j}$ otherwise. The negative of this minimum margin measures the remaining certified control authority.
\begin{theorem}
\label{thm:authority}
Under \Cref{ass:plant-regularity,ass:parameterization,ass:actuator,ass:data}, the set $\mathcal{K}_{N}^{a}(\mathbf{x},t)$ is nonempty at a query with finite margin \eqref{eq:margin} if and only if
\begin{equation}
\mathcal{A}_{N}^{a}:=-\bigl(d_{N}+\mathbf{b}_{N}^{\top}\mathbf{u}_{N}^{\mathrm{ext}}+\rho_{a,N}\bigr)\geq0.
\label{eq:authority}
\end{equation}
For a law $\boldsymbol{\kappa}_{N}$ satisfying \Cref{def:safety-informativity} and an initial state $\mathbf{x}(t_{0})\in\Omega(t_{0})$, we let $[t_{0},t_{\max})$ denote the maximal solution interval of the closed-loop system in $\mathcal{D}_{\Omega}$ under $\mathbf{u}_{c}=\boldsymbol{\kappa}_{N}(\mathbf{x},t)$. Then $\dot{W}\leq-\alpha(W)$ almost everywhere, while $W(t)\leq W(t_{0})$ and the output and realized-input constraints hold throughout $[t_{0},t_{\max})$.
\end{theorem}
\begin{proof}
The componentwise minimizer $\mathbf{u}_{N}^{\mathrm{ext}}$ of the affine margin \eqref{eq:margin} satisfies $\min_{\mathbf{u}_{c}\in\mathcal{U}_{\varepsilon}}\mathcal{M}_{N}^{a}=d_{N}+\rho_{a,N}+\mathbf{b}_{N}^{\top}\mathbf{u}_{N}^{\mathrm{ext}}=-\mathcal{A}_{N}^{a}$. Since the minimum is attained, \Cref{def:safety-informativity,def:pointwise-informativity} establish the feasibility condition \eqref{eq:authority}.

For each fixed admissible measurable disturbance, \Cref{ass:plant-regularity,ass:actuator} and the feedback regularity in \Cref{def:safety-informativity} ensure local well-posedness of the closed-loop system \eqref{eq:plant}, \eqref{eq:actuator}, and \eqref{eq:ea} in $(\mathbf{x},\mathbf{e}_{a})$. By \Cref{ass:data,ass:parameterization} and \Cref{lem:actuator-tube}, the actual parameter, remainder, and input error belong to the uncertainty sets in \eqref{eq:robust-dissipation}. Differentiating $W$ in \Cref{def:safety-coordinates} along \eqref{eq:channel} yields $\dot{W}=\mathbf{c}^{\top}\dot{\mathbf{e}}_{z}$, so $\dot{W}+\alpha(W)\leq\mathcal{M}_{N}^{a}\leq0$ almost everywhere by \eqref{eq:robust-dissipation} and \Cref{def:safety-informativity}. Integration implies $0\leq W(t)\leq W(t_{0})$ because $\alpha(W)\geq0$.

The storage bound and \Cref{def:safety-coordinates} imply $|z_{\ell}(t)|\leq Z:=\sup_{t\geq t_{0}}\|\mathbf{z}_{r}(t)\|+\sqrt{2W(t_{0})/\lambda_{\min}(\mathbf{P})}<\infty$, where $\lambda_{\min}(\mathbf{P})$ denotes the smallest eigenvalue. We set $\zeta:=(1+\exp(Z))^{-1}>0$; the inverse logit and width bound in \Cref{def:safety-coordinates} yield $\zeta\leq s_{\ell}(t)\leq1-\zeta$ and $\min\{h_{\ell}-\underline{y}_{\ell},\overline{y}_{\ell}-h_{\ell}\}\geq\zeta\underline{\varrho}_{\ell}>0$. Compactness of $\boldsymbol{\chi}(\mathcal{U}_{\varepsilon})$ in \Cref{ass:actuator} and \eqref{eq:tube} keep $\boldsymbol{\xi}=\boldsymbol{\chi}(\mathbf{u}_{c})+\mathbf{e}_{a}$ bounded. Hence, \eqref{eq:actuator} preserves the strict input bounds.
\end{proof}
Under the full-state continuation condition in \Cref{ass:plant-regularity}, bounded inputs preclude finite-time escape of the system state. If every finite accumulation point of the state-time trajectory belongs to $\mathcal{D}_{\Omega}$, the bounded system and actuator states permit local continuation, implying $t_{\max}=\infty$. For $\Omega(t)=\mathcal{C}(t)$, the positive output gaps exclude finite-time approach to the safe-set boundary, so the same continuation condition ensures all-time safety.

We define the static authority on the same command set as $\mathcal{A}_{N}^{\mathrm{stat}}:=-(d_{N}+\mathbf{b}_{N}^{\top}\mathbf{u}_{N}^{\mathrm{ext}})$. Then \eqref{eq:authority} becomes $\mathcal{A}_{N}^{a}=\mathcal{A}_{N}^{\mathrm{stat}}-\rho_{a,N}$: the support $\rho_{a,N}$ of the actuator-error set measures the reduction in certified authority due to the tracking-error bound.
\begin{corollary}
\label{cor:static}
If $\mathbf{e}_{a}(t_{0})=\mathbf{0}$ and $\bar{\omega}_{a,j}=0$ for all $j$, \Cref{lem:actuator-tube} yields $\mathbf{u}=\mathbf{u}_{c}$ and $\rho_{a,N}=0$. The command set, authority test, projection, and gate reduce to the static certificate on $\mathcal{U}_{\varepsilon}$.
\end{corollary}
In \Cref{cor:static}, the certified error bounds are zero. Retaining positive bounds in \Cref{lem:actuator-tube} can leave a positive reserve $\rho_{a,N}$ even when the realized actuator tracks perfectly.

When \eqref{eq:authority} holds, the minimum-deviation command sought in \Cref{prob:safety-synthesis} solves
\begin{equation}
\mathbf{u}_{s}=\argmin_{\mathbf{u}\in\mathcal{K}_{N}^{a}}
\tfrac{1}{2}\|\mathbf{u}-\mathbf{u}_{0}\|^{2}.
\label{eq:selector}
\end{equation}
The feasible set is the nonempty intersection of the compact rectangular command set $\mathcal{U}_{\varepsilon}$ with a halfspace, so strict convexity ensures a unique minimizer. We use $\operatorname{clip}_{\mathcal{U}_{\varepsilon}}$ for coordinatewise clipping and let $\lambda_{s}\geq0$ denote a halfspace multiplier. From the optimality conditions, one obtains $\mathbf{u}_{s}=\operatorname{clip}_{\mathcal{U}_{\varepsilon}}(\mathbf{u}_{0}-\lambda_{s}\mathbf{b}_{N})$. We take $\lambda_{s}=0$ when the clipped baseline is feasible; otherwise, we solve the scalar equation
\begin{equation}
d_{N}+\rho_{a,N}+\mathbf{b}_{N}^{\top}
\operatorname{clip}_{\mathcal{U}_{\varepsilon}}(\mathbf{u}_{0}-\lambda_{s}\mathbf{b}_{N})=0.
\label{eq:projection-root}
\end{equation}
The left-hand side of \eqref{eq:projection-root} is continuous, nonincreasing, and piecewise affine. Every coordinate with nonzero $b_{N,j}$ reaches the minimizing endpoint of its command interval at a finite multiplier, providing a finite root bracket whenever \eqref{eq:authority} holds. Multiple roots can represent the same unique command. Bisection approximates a root, and the returned command must satisfy \eqref{eq:margin} with a certified sign in finite precision.

Using this projection as a feedback law also requires regularity. Suppose $\mathbf{u}_{0}$, $\mathbf{b}_{N}$, $d_{N}$, and $\rho_{a,N}$ are locally Lipschitz on a neighborhood where \eqref{eq:margin} is finite. At a feasible query, linear independence of the active constraint normals is sufficient for local Lipschitz continuity of $\mathbf{u}_{s}$. Here, an active constraint is an inequality satisfied with equality. The corresponding optimality systems remain nonsingular locally, and their finitely many locally Lipschitz solution formulas join continuously at changes of active set. The existence of commands that satisfy the safety inequality with uniformly positive slack does not, by itself, guarantee local Lipschitz continuity of the projected command $\mathbf{u}_{s}$.

To determine the largest certified fraction along $\mathbf{u}_{c,\lambda}$, we let $m_{s}:=\mathfrak{m}_{N}^{a}(\mathbf{u}_{s})\leq0$ and $m_{d}:=\mathfrak{m}_{N}^{a}(\mathbf{u}_{d})$ denote the endpoint margins. The gate then takes the form
\begin{equation}
\lambda_{N,a}^{\star}=
\begin{cases}
1,&m_{d}\leq0,\\
-m_{s}/(m_{d}-m_{s}),&m_{d}>0.
\end{cases}
\label{eq:gate}
\end{equation}
Since both endpoint commands belong to the convex command set $\mathcal{U}_{\varepsilon}$, every $\mathbf{u}_{c,\lambda}$ with $\lambda\in[0,1]$ satisfies the command limits. By \eqref{eq:margin}, the margin along the segment is $\mathfrak{m}_{N}^{a}(\mathbf{u}_{c,\lambda})=(1-\lambda)m_{s}+\lambda m_{d}$. If $m_{d}\leq0$, the desired command is certified and $\lambda=1$ is admissible. For $m_{d}>0$, the certificate requires $\lambda\leq-m_{s}/(m_{d}-m_{s})$, establishing \eqref{eq:gate}.

When $m_{d}>0$, a positive admitted fraction is possible if and only if the safe endpoint has strict slack, $m_{s}<0$. For a desired command with $m_{d}>0$, choosing $\mathbf{u}_{0}=\mathbf{u}_{d}$ in \eqref{eq:selector} places the projection on the halfspace boundary, where $m_{s}=0$ and \eqref{eq:gate} returns $\lambda_{N,a}^{\star}=0$.

Suppose $m_{s}$ and $m_{d}$ are locally Lipschitz in $(\mathbf{x},t)$ and $m_{s}\leq-\mu$ on a neighborhood for some $\mu>0$. On the $m_{d}>0$ branch, the denominator satisfies $m_{d}-m_{s}\geq\mu$; at $m_{d}=0$, both formulas in \eqref{eq:gate} equal one. The branch formulas define a locally Lipschitz gate. With locally Lipschitz endpoint commands $\mathbf{u}_{s}$ and $\mathbf{u}_{d}$, the applied command satisfies the feedback regularity required by \Cref{thm:authority}; the actuator-error bounds must also remain valid for the resulting command. Discontinuities in the gate and applied command can occur at $m_{s}=m_{d}=0$, where \eqref{eq:gate} assigns $\lambda_{N,a}^{\star}=1$.
\begin{remark}
\label{rem:regularity-obstruction}
The decay certificate imposes stronger requirements than output safety. We consider scalar error, input, and disturbance $e,u,d\in\mathbb{R}$, with $\dot{e}=u+d$, $|d|\leq r$ for $r>0$, $W=e^{2}/2$, and $\alpha(W)=2kW$ for $k>0$. The robust margin $eu+r|e|+ke^{2}$ requires $u\leq-r-ke$ for $e>0$ and $u\geq r-ke$ for $e<0$. A sufficiently wide input interval permits pointwise feasibility near zero, but the incompatible one-sided bounds preclude a continuous selector through zero. Pointwise informativity in \Cref{def:pointwise-informativity} does not guarantee the locally Lipschitz feedback selection required by \Cref{def:safety-informativity}.
\end{remark}
\section{Numerical Results and Discussions}
We consider a vehicle moving in the horizontal--vertical plane under a structured wind field and bounded turbulence. The position $\mathbf{x}=[p_x,h]^{\top}$ evolves according to
\begin{align}
    \dot{p}_{x}=&u_x+w_x(p_x,h,t)+d_x,\\
    \dot{h}=&u_h+w_h(p_x,h,t)+d_h,
    \label{eq:vehicle-model}
\end{align}
where $\mathbf{u}=[u_x,u_h]^{\top}$ is the realized velocity input and $\|[d_x,d_h]^{\top}\|\in[0,0.04]$~m/s. The simulated wind field is
\begin{align}
    w_x=&0.35\sin(0.006p_x+0.09t)
    +0.15\tanh\left(\frac{h-100}{20}\right),\\
    w_h=&0.25\cos(0.008p_x-0.07t)
    +0.12\exp\left(-\frac{(h-105)^2}{300}\right).
    \label{eq:vehicle-wind}
\end{align}
The feature functions are known to the controller; the wind coefficients specify the simulated system and are treated as unknown in the finite-data certificate. During data acquisition and pointwise evaluation, the reference is $p_{x,r}(t)=13t$ and $h_r(t)=100+6\sin(0.08t)$. With $e_x:=p_x-p_{x,r}$ and $e_h:=h-h_r$, the tracking constraints are $e_x\in(-20,30)$~m and $e_h\in(-12,18)$~m.
The physical velocity limits are $u_x\in(8,18)$~m/s and $u_h\in(-3,2)$~m/s. Commands are restricted to the inner command set $\mathcal{U}_{\varepsilon}=[9,17]\times[-2.5,1.5]$~m/s.

Let $\mathbf{x}_{r}(t):=[p_{x,r}(t),h_r(t)]^{\top}$ denote the reference position. At a fixed time $t$, the tracking constraints define the shaded rectangle in \Cref{fig:vehicle-test-scenario}. The left and right boundaries of $\mathcal{C}(t)$ are $p_x=p_{x,r}(t)-20$ and $p_x=p_{x,r}(t)+30$; the lower and upper boundaries are $h=h_r(t)-12$ and $h=h_r(t)+18$, respectively. The rectangle translates with the reference position, leaving more admissible displacement ahead of and above the reference. The dashed curve traces the reference position over time. The black diamond and blue disk mark the reference and vehicle positions at the same instant, and the directed error arrows connect the two positions through their horizontal and vertical displacements.

The blue velocity arrows represent the controlled components $u_x$ and $u_h$. Wind and turbulence contribute additional velocity through \eqref{eq:vehicle-model}, so the controller must use the available input authority both to follow the reference and to compensate for the environmental terms. The illustrated vehicle lies ahead of and above the reference, corresponding to $e_x,e_h\in(0,\infty)$.
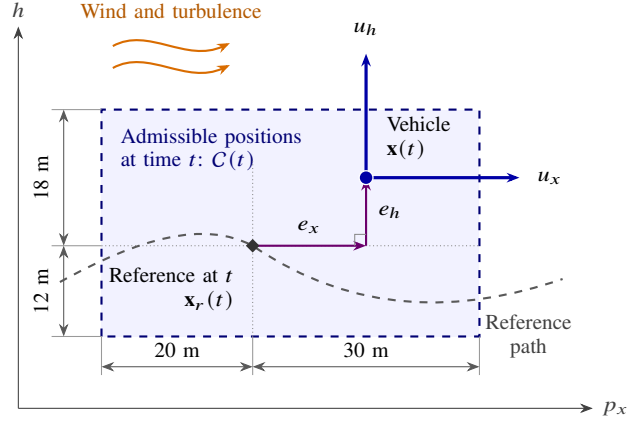
\begin{figure}[t]
    \centering
    \begin{tikzpicture}[
    x=1cm,y=1cm,font=\footnotesize,
    >={Stealth[length=1.7mm]},
    scenario dimension/.style={<->,thin,draw=black!65},
    scenario label/.style={fill=none,inner sep=1.3pt},
    scenario error/.style={->,thick,draw=violet!85!black},
    scenario input/.style={->,very thick,draw=blue!65!black},
    scenario disturbance/.style={->,thick,draw=orange!85!black}
]
    \coordinate (reference) at (3.3,2.6);
    \coordinate (vehicle) at (4.8,3.5);
    \coordinate (projection) at (4.8,2.6);

    \fill[blue!5] (1.3,1.4) rectangle (6.3,4.4);
    \draw[dashed,thick,blue!45!black] (1.3,1.4) rectangle (6.3,4.4);
    \node[anchor=north west,align=left,text=blue!45!black] at (1.43,4.28)
        {Admissible positions\\at time $t$: $\mathcal{C}(t)$};

    \draw[->,black!75] (0.2,0.45) -- (7.8,0.45)
        node[right] {$p_x$};
    \draw[->,black!75] (0.2,0.45) -- (0.2,5.5)
        node[above] {$h$};

    \draw[densely dotted,black!45] (3.3,1.4) -- (3.3,3.72);
    \draw[densely dotted,black!45] (1.3,2.6) -- (6.3,2.6);
    \draw[dashed,thick,black!65]
        (0.75,2.12) .. controls (1.8,2.7) and (2.65,2.94) .. (reference)
        .. controls (4.65,1.68) and (6.05,1.7) .. (7.4,2.16);
    \node[anchor=north,align=center,text=black!75] at (6.95,1.81)
        {Reference\\path};

    \draw[scenario error] (reference) -- (projection)
        node[midway,above=3pt,scenario label] {$e_x$};
    \draw[scenario error] (projection) -- (vehicle)
        node[midway,right=3pt,scenario label] {$e_h$};
    \draw[black!45,thin] (4.65,2.6) -- (4.65,2.75) -- (4.8,2.75);
    \fill[black!80]
        (3.3,2.69) -- (3.39,2.6) -- (3.3,2.51) -- (3.21,2.6) -- cycle;
    \node[anchor=north east,align=right] at (3.18,2.43)
        {Reference at $t$\\$\mathbf{x}_{r}(t)$};

    \draw[scenario input] (vehicle) -- (6.83,3.5)
        node[right=3pt] {$u_x$};
    \draw[scenario input] (vehicle) -- (4.8,5.14)
        node[above=3pt] {$u_h$};
    \filldraw[fill=blue!65!black,draw=white,line width=0.5pt]
        (vehicle) circle (2.6pt);
    \node[anchor=south west,align=left] at (4.95,3.62)
        {Vehicle\\$\mathbf{x}(t)$};

    \draw[black!40,thin] (1.3,1.4) -- (1.3,0.89);
    \draw[black!40,thin] (3.3,1.4) -- (3.3,0.89);
    \draw[black!40,thin] (6.3,1.4) -- (6.3,0.89);
    \draw[scenario dimension] (1.3,1) -- (3.3,1)
        node[midway,above=2pt,scenario label] {$20$~m};
    \draw[scenario dimension] (3.3,1) -- (6.3,1)
        node[midway,above=2pt,scenario label] {$30$~m};
    \draw[black!40,thin] (1.3,1.4) -- (0.64,1.4);
    \draw[black!40,thin] (1.3,2.6) -- (0.64,2.6);
    \draw[black!40,thin] (1.3,4.4) -- (0.64,4.4);
    \draw[scenario dimension] (0.8,1.4) -- (0.8,2.6);
    \draw[scenario dimension] (0.8,2.6) -- (0.8,4.4);
 
    \node[rotate=90,scenario label] at (0.5,2) {$12$~m};
    \node[rotate=90,scenario label] at (0.5,3.5) {$18$~m};

    \draw[scenario disturbance]
        (1.45,4.95) .. controls (1.92,5.21) and (2.42,4.72) .. (3,4.99);
    \draw[scenario disturbance]
        (1.45,5.24) .. controls (1.92,5.5) and (2.42,5.01) .. (3,5.28);
    \node[anchor=south,text=orange!70!black] at (2.22,5.49)
        {Wind and turbulence};

\end{tikzpicture}
    \caption{Position tracking with asymmetric error bounds.}
    \label{fig:vehicle-test-scenario}
\end{figure}

Applying \Cref{def:safety-coordinates} to the tracking errors yields
\begin{equation}
    \mathbf{z}=\begin{bmatrix}
    \log\left(\dfrac{e_x+20}{30-e_x}\right)\\
    \log\left(\dfrac{e_h+12}{18-e_h}\right)
    \end{bmatrix},~~
    \mathbf{z}_{r}=\begin{bmatrix}\log(2/3)\\\log(2/3)\end{bmatrix}.
\end{equation}
We use $\mathbf{P}=\operatorname{diag}(\mathrm{e},\mathrm{e}^{-1})$ and $\bar{W}=0.75$, where $\mathrm{e}$ is Euler's number. The evaluation region $\Omega(t):=\{\mathbf{x}:W(\mathbf{x},t)\in[0,\bar{W}]\}$ is fixed in tracking-error coordinates and moves with the reference in physical coordinates. Define
\begin{equation}
    \zeta_x:=\frac{50}{(e_x+20)(30-e_x)},~~
    \zeta_h:=\frac{30}{(e_h+12)(18-e_h)}.
\end{equation}
The nominal transformed terms are $\mathbf{B}=\operatorname{diag}(\zeta_x,\zeta_h)$ and $\mathbf{a}=-\mathbf{B}[\dot{p}_{x,r},\dot{h}_{r}]^{\top}$. The wind contribution has the representation $\boldsymbol{\Theta}^{\star\top}\boldsymbol{\phi}$, with
\begin{equation}
\boldsymbol{\phi}(\mathbf{x},t)=
 \begin{bmatrix}
  \zeta_x\sin(0.006p_x+0.09t)\\
  \zeta_x\tanh((h-100)/20)\\
  \zeta_h\cos(0.008p_x-0.07t)\\
  \zeta_h\exp(-(h-105)^2/300)
 \end{bmatrix}.
\end{equation}
The rows of $\boldsymbol{\Theta}^{\star}\in\mathbb{R}^{4\times2}$ are $[0.35,0]$, $[0.15,0]$, $[0,0.25]$, and $[0,0.12]$. For the actuator-error model in \Cref{ass:actuator}, we use $\boldsymbol{\Lambda}_{a}=\operatorname{diag}(2.5,4)$ and $\boldsymbol{\omega}_{a}\in[-0.5,0.5]\times[-0.2,0.2]$.

The initial acquisition campaign comprises five 30-s flights, starting from tracking errors $(-4,-3)$, $(-2,2)$, $(0,0)$, $(2,-2)$, and $(4,3)$~m. Bounded velocity commands provide sinusoidal excitation. Sampling the state, realized input, and transformed derivative every $0.2$~s produces 755 measurements. Two records retain $N=40$ measurements each: the consecutive record $\mathcal{D}_{\mathrm{c}}$ contains the first 40 samples of the first flight, and the uniform record $\mathcal{D}_{\mathrm{u}}$ contains eight approximately uniformly spaced samples from each flight.

Later, a broader acquisition campaign comprises 72 admissible 20-s flights with varied initial errors, excitation phases, and starting times $t_0\in\{0,15,30,45\}$~s. The resulting pool contains 7272 samples. From the samples in $\Omega(t)$, the certificate-directed record $\mathcal{D}_{\mathrm{s}}$ retains 40 measurements, with at most two from any flight. Selection first acquires the four feature directions and then favors weakly represented directions, weighted by $\|\mathbf{P}(\mathbf{z}-\mathbf{z}_{r})\|$. The retained sample count is common to all three records; the selected record draws on a larger acquisition campaign.
\begin{figure}[!t]
    \centering
    \captionsetup[subfigure]{font=footnotesize,skip=2pt}
    \begin{subfigure}[t]{.49\linewidth}
    \centering
    \includegraphics[width=\linewidth]{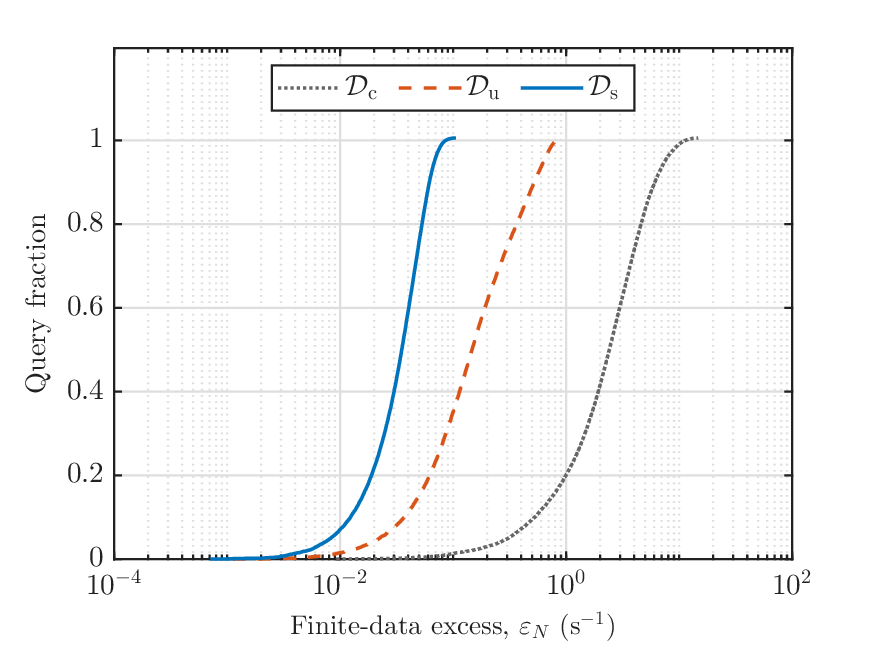}
    \caption{CDF of finite-data excess.}
    \label{fig:final_excess}
    \end{subfigure}\hfill
    \begin{subfigure}[t]{.49\linewidth}
    \centering
    \includegraphics[width=\linewidth]{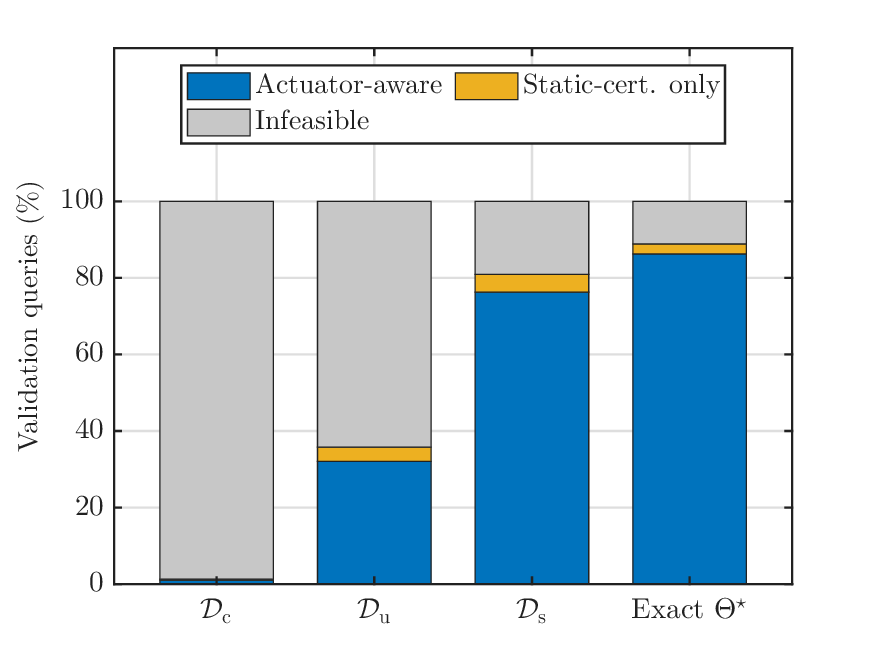}
    \caption{Pointwise certificate coverage.}
    \label{fig:final_coverage}
    \end{subfigure}
    \par\smallskip
    \begin{subfigure}[t]{.49\linewidth}
    \centering
    \includegraphics[width=\linewidth]{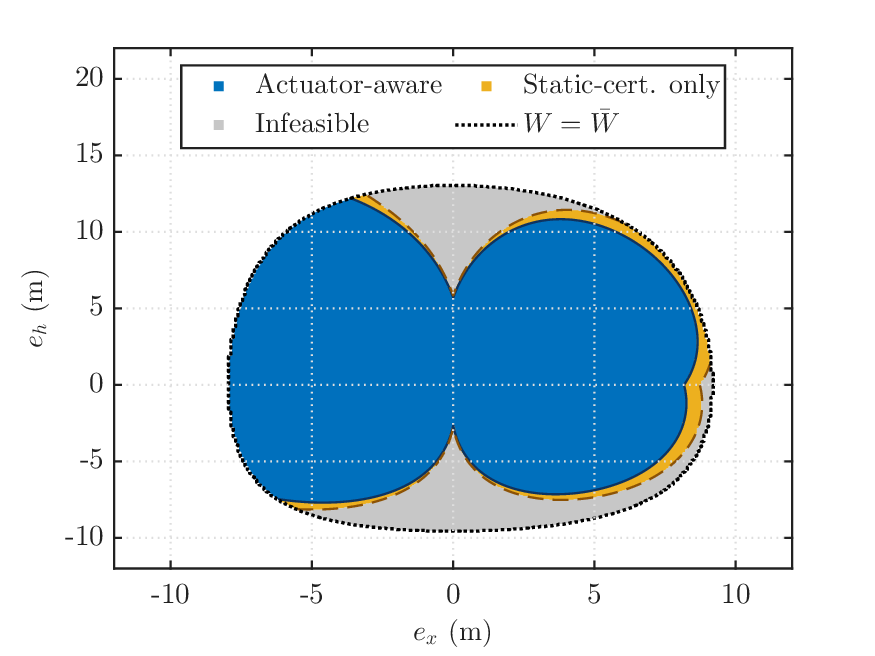}
    \caption{Certificate regions for $\mathcal{D}_s$.}
    \label{fig:final_map}
    \end{subfigure}
    \caption{Effect of dataset selection on safety informativity.}
    \label{fig:ini_1}
\end{figure}
\begin{figure}[!t]
    \centering
    \captionsetup[subfigure]{font=footnotesize,skip=2pt}
    \begin{subfigure}[t]{.49\linewidth}
    \centering
    \includegraphics[width=\linewidth]{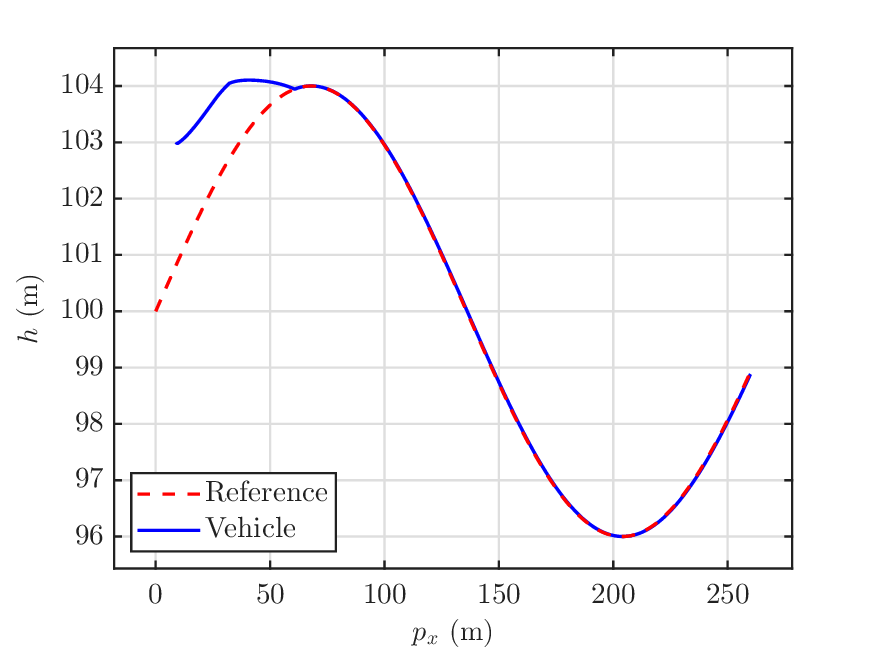}
    \caption{Trajectories.}
    \label{fig:2a}
    \end{subfigure}\hfill
    \begin{subfigure}[t]{.49\linewidth}
    \centering
    \includegraphics[width=\linewidth]{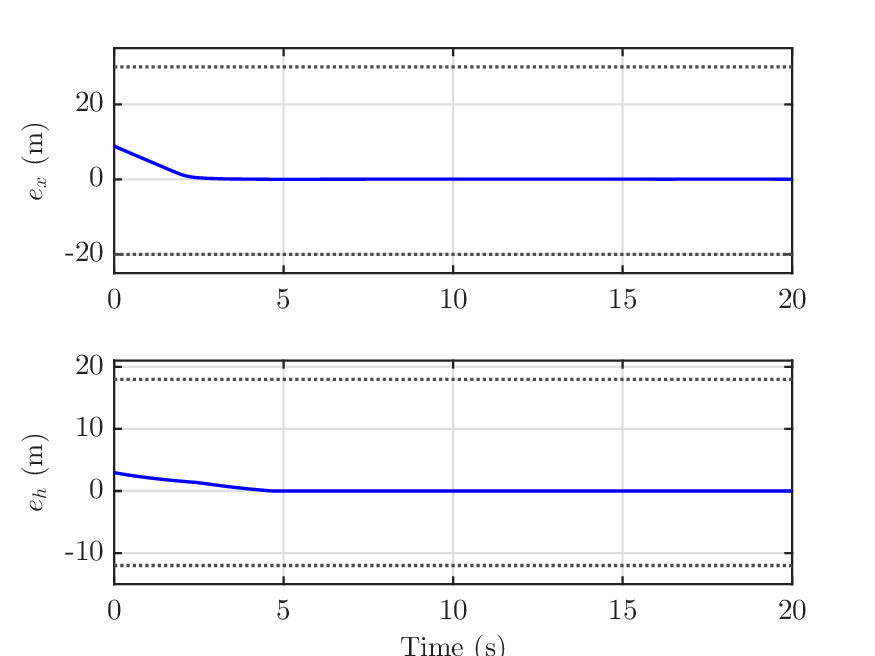}
    \caption{Errors.}
    \label{fig:2b}
    \end{subfigure}
    \par\smallskip
    \begin{subfigure}[t]{.49\linewidth}
    \centering
    \includegraphics[width=\linewidth]{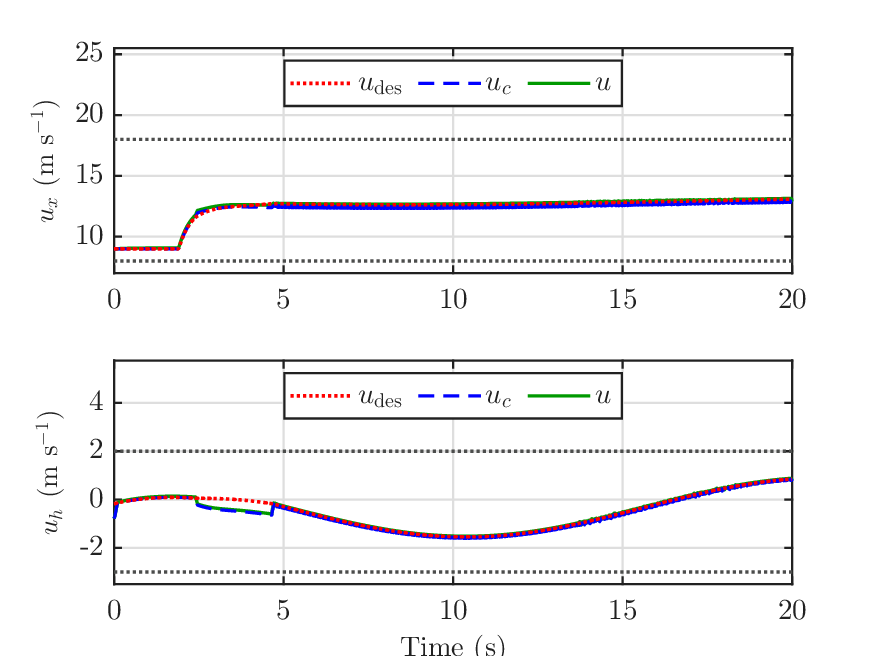}
    \caption{Control inputs.}
    \label{fig:2c}
    \end{subfigure}\hfill
    \begin{subfigure}[t]{.49\linewidth}
    \centering
    \includegraphics[width=\linewidth]{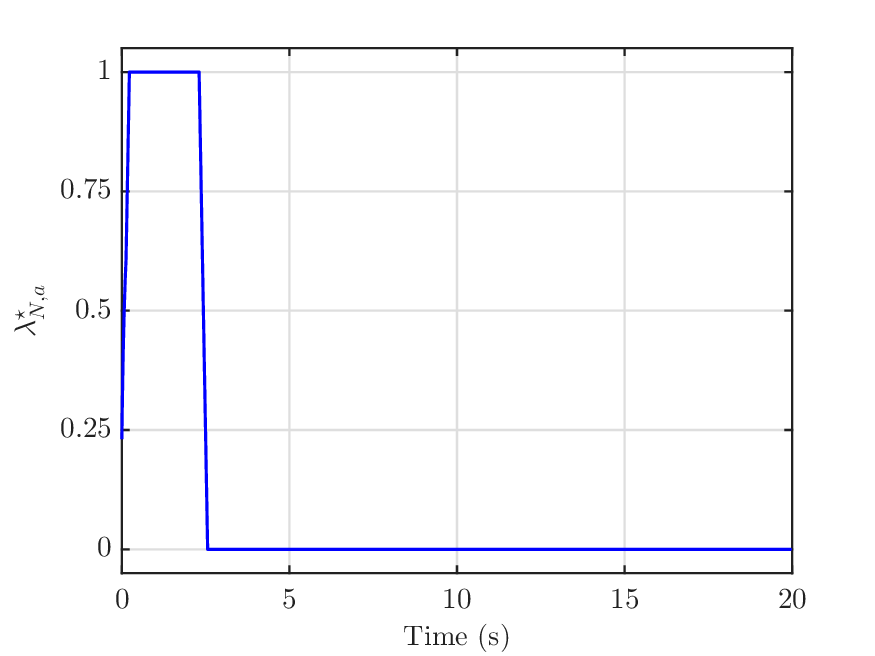}
    \caption{Blending coefficient.}
    \label{fig:2d}
    \end{subfigure}
    \caption{Closed-loop performance of the data-driven controller.}
    \label{fig:2}
\end{figure}

For each record, the sampled features and residuals form $\boldsymbol{\Phi}_{N}$ and $\mathbf{R}_{N}$ as in \Cref{ass:data}. Since $\mathbf{z}_{r}$ is constant, the residual at sample $k$ is $\mathbf{r}_{k}=\dot{\mathbf{z}}_{k}^{\mathrm{meas}}-\mathbf{a}_{k}-\mathbf{B}_{k}\mathbf{u}_{k}$. The consistency set then becomes
\begin{equation}
 \mathcal{T}_{N}=\{\boldsymbol{\Theta}\in\mathbb{R}^{4\times2}:
 \|\mathbf{R}_{N}-\boldsymbol{\Theta}^{\top}\boldsymbol{\Phi}_{N}\|_{\mathrm{F}}\in[0,\eta_N]\}.
 \label{eq:fig1_parameter_set}
\end{equation}
A transformed measurement-error bound of $0.004$~s$^{-1}$ and a transformed turbulence bound of $\bar{d}\approx0.01785$~s$^{-1}$ yield $\eta_N=\sqrt{40}(\bar{d}+0.004)\approx0.13819$~s$^{-1}$. Each record has $\operatorname{rank}(\boldsymbol{\Phi}_{N})=4$ and a nonempty consistency set, so \Cref{thm:data-support} applies to every query feature.

We evaluate all three certificates on a common set $\mathcal{Q}_{\mathrm{val}}$ of 11,241 state--time queries. A $41\times41$ grid restricted to the unit disk provides 1249 points $\mathbf{q}$, which are mapped through $\mathbf{z}=\mathbf{z}_{r}+\sqrt{2\bar{W}}\mathbf{P}^{-1/2}\mathbf{q}$ and the inverse logit at each of the nine times $t=0,7.5,\ldots,60$~s. The validation queries are used only for certificate evaluation. All comparisons use $\alpha(W)=0.8W$ and the asymptotic physical-input tube $\boldsymbol{\varepsilon}_{a,\infty}=[0.2,0.05]^{\top}$~m/s.

The finite-data excess is the difference between the support in \eqref{eq:support} and the true wind contribution $\mathbf{c}^{\top}\boldsymbol{\Theta}^{\star\top}\boldsymbol{\phi}$. In \Cref{fig:final_excess}, the empirical cumulative distribution records the fraction of validation queries whose excess is at most the horizontal-axis value. The reported median excess decreases from $2.423$~s$^{-1}$ for $\mathcal{D}_{\mathrm{c}}$ to $0.1477$~s$^{-1}$ for $\mathcal{D}_{\mathrm{u}}$ and $0.0353$~s$^{-1}$ for $\mathcal{D}_{\mathrm{s}}$. Full row rank ensures finite support, while the factors $\|\boldsymbol{\Phi}_{N}^{\dagger}\boldsymbol{\phi}\|$ and $\chi_N$ in \eqref{eq:support} determine the uncertainty penalty. Measurements distributed across flights provide a tighter certificate here, and selection based on underrepresented feature directions reduces the penalty further.

The actuator-aware certificate holds at $1.03\%$, $32.09\%$, and $76.27\%$ of the validation queries for $\mathcal{D}_{\mathrm{c}}$, $\mathcal{D}_{\mathrm{u}}$, and $\mathcal{D}_{\mathrm{s}}$, respectively; see \Cref{fig:final_coverage}. These percentages measure coverage of the sampled state--time set. Under the authority condition \eqref{eq:authority}, reducing the finite-data support leaves more of the bounded command authority available for storage decrease and actuator uncertainty. The exact-parameter benchmark certifies $86.26\%$ of the same queries. The exact-parameter benchmark's $9.99$ percentage-point advantage over $\mathcal{D}_{\mathrm{s}}$ quantifies the remaining finite-data conservatism for this comparison. The benchmark itself remains limited by the prescribed command bounds, disturbance reserve, actuator tube, and storage inequality.

\Cref{fig:final_map} locates the certificate classes for $\mathcal{D}_{\mathrm{s}}$ at $t=0$ within the boundary $W=\bar{W}$. The blue region satisfies the actuator-aware authority condition. In the yellow region, the static authority is nonnegative but smaller than the actuator reserve $\rho_{a,N}$, so the static certificate is feasible while the actuator-aware certificate is infeasible. The gray region has negative static authority. The relation $\mathcal{A}_{N}^{a}=\mathcal{A}_{N}^{\mathrm{stat}}-\rho_{a,N}$ accounts for the change in the feasible region. Pointwise coverage characterizes command availability at the sampled queries; the closed-loop experiment below assesses feasibility along a particular trajectory.

For the closed-loop experiment, we retain $\mathcal{D}_{\mathrm{s}}$ and use the reference $p_{x,r}(t)=13t$, $h_r(t)=100+4\sin(0.3t)$. Over $t\in[0,20]$~s, the initial tracking errors are $(e_x(0),e_h(0))=(8.9,3)$~m, corresponding to $W(0)=0.7349\in[0,\bar{W})$. The actuator error starts at zero and is driven by the constant forcing $\boldsymbol{\omega}_{a}=[0.5,0.2]^{\top}$. Projection of the reference velocity $\mathbf{u}_{0}=[13,1.2\cos(0.3t)]^{\top}$ onto the actuator-aware command set, using \eqref{eq:selector}--\eqref{eq:projection-root}, determines $\mathbf{u}_{s}$.

The desired tracking command uses the least-squares wind estimate:
\begin{equation}
 \widehat{\boldsymbol{\Theta}}^{\top}
 =\mathbf{R}_{N}\boldsymbol{\Phi}_{N}^{\dagger}.
\end{equation}
With the feature and residual matrices of $\mathcal{D}_{\mathrm{s}}$, we set
\begin{equation}
 \mathbf{u}_{d}=\operatorname{clip}_{\mathcal{U}_{\varepsilon}}
 \left[\mathbf{B}^{-1}\left(-0.9\mathbf{c}-\mathbf{a}
 -\widehat{\boldsymbol{\Theta}}^{\top}\boldsymbol{\phi}\right)\right].
 \label{eq:fig6_desired}
\end{equation}
The actuator-aware safety certificate requires $\mathcal{M}_{N}^{a}(\mathbf{x},\mathbf{u}_{c},t)\in(-\infty,0]$. We evaluate the data contribution using \eqref{eq:support} over the full consistency set $\mathcal{T}_{N}$ and incorporate the disturbance and actuator-error reserves through \eqref{eq:margin}. The gate in \eqref{eq:gate} selects the largest fraction of $\mathbf{u}_{d}$ for which $\mathbf{u}_{c}=(1-\lambda_{N,a}^{\star})\mathbf{u}_{s}+\lambda_{N,a}^{\star}\mathbf{u}_{d}$ satisfies this condition.

The vehicle approaches the reference from the initial position offset and follows the altitude variation in \Cref{fig:2a}. Both tracking errors remain within their asymmetric bounds in \Cref{fig:2b}. The reported entry times into $[-0.2,0.2]$~m are $3.05$~s for $e_x$ and $4.25$~s for $e_h$, with both errors remaining in that interval for the rest of the simulation. The actuator-aware margin is nonpositive to numerical precision at every simulated command update, and the maximum sampled storage is $0.7349$. These sampled results agree with the dissipation condition of \Cref{thm:authority} and place the simulated trajectory within the evaluation region.

\Cref{fig:2c} compares desired, commanded, and realized velocities. The gate modifies $\mathbf{u}_{d}$ when the desired command violates the certificate, while actuator tracking error accounts for the difference between $\mathbf{u}_{c}$ and $\mathbf{u}$. The reported realized-input ranges are approximately $[9,14]$~m/s horizontally and $[-1.856,1.044]$~m/s vertically, within the physical limits. The uncertainty reserve accounts for velocity mismatch in the storage inequality; the range of the actuator map enforces the physical input bounds. The displayed velocities use an 11-sample moving mean, whereas feasibility and bound checks use the unsmoothed samples at $0.005$-s intervals.

The desired command is fully admitted between approximately $0.1$ and $2.44$~s in \Cref{fig:2d}. Later, the gate approaches zero as the projected endpoint has negligible certificate slack and the desired endpoint violates the certificate. The fraction $-m_s/(m_d-m_s)$ then approaches zero, and the commanded input follows $\mathbf{u}_{s}$. The projected reference command continues to satisfy the storage inequality as the tracking errors decrease. A 51-sample moving mean is used only to display the gate; the implementation uses the unsmoothed coefficient. The displayed transitions do not establish continuity at zero-slack endpoints. The reported step-refinement check preserves sampled certificate feasibility and changes each settling time by $0.0025$~s when the integration step is halved.

\section{Conclusions and Future Work}
We developed a finite-data framework for certifying nonlinear output safety under dynamic asymmetric actuation. The measured regressor span characterized the queries admitting finite support for nonzero safety functionals without prior parameter bounds. Combining data support with a certified actuator-error tube yielded a pointwise feasibility condition, a minimum-deviation command, and a maximal gate along a prescribed command segment. Under the stated regularity and continuation conditions, the feedback preserved output and input constraints within the certification domain. The vehicle study indicated that measurement selection from a broader flight campaign tightened certificates and increased coverage despite equal record sizes and full row rank. The reported response maintained tracking and velocity bounds, with nonpositive certificate margins at sampled command updates. The study motivated acquisition strategies balancing certificate improvement against sampling effort and augmented-state certificates that retained actuator dynamics to reduce conservatism.

    	\bibliographystyle{IEEEtran}
	\bibliography{references}

\begin{thebibliography}{10}
\providecommand{\url}[1]{#1}
\csname url@samestyle\endcsname
\providecommand{\newblock}{\relax}
\providecommand{\bibinfo}[2]{#2}
\providecommand{\BIBentrySTDinterwordspacing}{\spaceskip=0pt\relax}
\providecommand{\BIBentryALTinterwordstretchfactor}{4}
\providecommand{\BIBentryALTinterwordspacing}{\spaceskip=\fontdimen2\font plus
\BIBentryALTinterwordstretchfactor\fontdimen3\font minus
  \fontdimen4\font\relax}
\providecommand{\BIBforeignlanguage}[2]{{%
\expandafter\ifx\csname l@#1\endcsname\relax
\typeout{** WARNING: IEEEtran.bst: No hyphenation pattern has been}%
\typeout{** loaded for the language `#1'. Using the pattern for}%
\typeout{** the default language instead.}%
\else
\language=\csname l@#1\endcsname
\fi
#2}}
\providecommand{\BIBdecl}{\relax}
\BIBdecl

\bibitem{Blanchini1999Survey}
F.~Blanchini, ``Set invariance in control,'' \emph{Automatica}, vol.~35,
  no.~11, pp. 1747--1767, 1999.

\bibitem{Ames2017CBF}
A.~D. Ames, X.~Xu, J.~W. Grizzle, and P.~Tabuada, ``Control barrier function
  based quadratic programs for safety critical systems,'' \emph{IEEE
  Transactions on Automatic Control}, vol.~62, no.~8, pp. 3861--3876, 2017.

\bibitem{Wabersich2021PSF}
K.~P. Wabersich and M.~N. Zeilinger, ``A predictive safety filter for
  learning-based control of constrained nonlinear dynamical systems,''
  \emph{Automatica}, vol. 129, p. 109597, 2021.

\bibitem{Xiao2022HOCBF}
W.~Xiao and C.~Belta, ``High-order control barrier functions,'' \emph{IEEE
  Transactions on Automatic Control}, vol.~67, no.~7, pp. 3655--3662, 2022.

\bibitem{Xiao2022Feasibility}
W.~Xiao, C.~A. Belta, and C.~G. Cassandras, ``Sufficient conditions for
  feasibility of optimal control problems using control barrier functions,''
  \emph{Automatica}, vol. 135, p. 109960, 2022.

\bibitem{Jagtap2020GPCBF}
P.~Jagtap, G.~J. Pappas, and M.~Zamani, ``Control barrier functions for unknown
  nonlinear systems using gaussian processes,'' in \emph{2020 59th IEEE
  Conference on Decision and Control (CDC)}, 2020, pp. 3699--3704.

\bibitem{Dhiman2023BayesianCBF}
V.~Dhiman, M.~J. Khojasteh, M.~Franceschetti, and N.~Atanasov, ``Control
  barriers in {Bayesian} learning of system dynamics,'' \emph{IEEE Transactions
  on Automatic Control}, vol.~68, no.~1, pp. 214--229, 2023.

\bibitem{Choi2020RLCBFCLF}
J.~J. Choi, F.~Casta{\~n}eda, C.~J. Tomlin, and K.~Sreenath, ``Reinforcement
  learning for safety-critical control under model uncertainty, using control
  lyapunov functions and control barrier functions,'' in \emph{Proceedings of
  Robotics: Science and Systems}, Corvallis, OR, USA, July 2020.

\bibitem{HsuTsukamoto2025Conformal}
T.-W. Hsu and H.~Tsukamoto, ``Statistical guarantees in data-driven nonlinear
  control: Conformal robustness for stability and safety,'' 2025.

\bibitem{Lopez2021RobustAdaptiveCBF}
B.~T. Lopez, J.-J.~E. Slotine, and J.~P. How, ``Robust adaptive control barrier
  functions: An adaptive and data-driven approach to safety,'' \emph{IEEE
  Control Systems Letters}, vol.~5, no.~3, pp. 1031--1036, 2021.

\bibitem{Das2025UncertaintyCBF}
E.~Da{\c{s}} and J.~W. Burdick, ``Robust control barrier functions using
  uncertainty estimation with application to mobile robots,'' \emph{IEEE
  Transactions on Automatic Control}, vol.~70, no.~7, pp. 4766--4773, 2025.

\bibitem{VanWaarde2020Informativity}
H.~J. van Waarde, J.~Eising, H.~L. Trentelman, and M.~K. Camlibel, ``Data
  informativity: A new perspective on data-driven analysis and control,''
  \emph{IEEE Transactions on Automatic Control}, vol.~65, no.~11, pp.
  4753--4768, 2020.

\bibitem{VanWaarde2023CSM}
H.~J. van Waarde, J.~Eising, M.~K. Camlibel, and H.~L. Trentelman, ``The
  informativity approach: To data-driven analysis and control,'' \emph{IEEE
  Control Systems Magazine}, vol.~43, no.~6, pp. 32--66, 2023.

\bibitem{DePersis2020Formulas}
C.~De~Persis and P.~Tesi, ``Formulas for data-driven control: Stabilization,
  optimality, and robustness,'' \emph{IEEE Transactions on Automatic Control},
  vol.~65, no.~3, pp. 909--924, 2020.

\bibitem{VanWaarde2022MatrixSLemma}
H.~J. van Waarde, M.~K. Camlibel, and M.~Mesbahi, ``From noisy data to feedback
  controllers: Nonconservative design via a matrix {S}-lemma,'' \emph{IEEE
  Transactions on Automatic Control}, vol.~67, no.~1, pp. 162--175, 2022.

\bibitem{Bisoffi2022Petersen}
A.~Bisoffi, C.~De~Persis, and P.~Tesi, ``Data-driven control via {Petersen}'s
  lemma,'' \emph{Automatica}, vol. 145, p. 110537, 2022.

\bibitem{Martin2023Survey}
T.~Martin, T.~B. Sch{\"o}n, and F.~Allg{\"o}wer, ``Guarantees for data-driven
  control of nonlinear systems using semidefinite programming: A survey,''
  2023.

\bibitem{Strong2025DeterministicCBF}
A.~K. Strong, A.~Kashani, C.~Danielson, and L.~Bridgeman, ``Learning control
  barrier functions with deterministic safety guarantees,'' 2025.

\bibitem{Ames2020IntegralCBF}
A.~D. Ames, G.~Notomista, Y.~Wardi, and M.~Egerstedt, ``Integral control
  barrier functions for dynamically defined control laws,'' \emph{IEEE Control
  Systems Letters}, vol.~5, no.~3, pp. 887--892, 2021.

\bibitem{Seiler2021IQC}
P.~Seiler, M.~Jankovic, and E.~Hellstrom, ``Control barrier functions with
  unmodeled input dynamics using integral quadratic constraints,'' \emph{IEEE
  Control Systems Letters}, vol.~6, pp. 1664--1669, 2022.

\bibitem{Kumar2026APC}
S.~Kumar, S.~R. Kumar, and A.~Sinha, ``Admissibility-preserving control for
  strict-feedback nonlinear systems with asymmetric actuator constraints,''
  arXiv preprint arXiv:2608.15375, 2026.

\bibitem{SinhaKumar2026AsymmetricConsensus}
A.~Sinha and S.~R. Kumar, ``Distributed safe consensus under asymmetric input
  and time-varying output constraints,'' arXiv preprint arXiv:2606.16116, 2026.

\bibitem{Sinha2026DirectedSafeCoordination}
A.~Sinha, L.~Gopikannan, and S.~R. Kumar, ``Networked admissibility-preserving
  control for directed safe coordination,'' arXiv preprint arXiv:2609.09384,
  2026.

\end{thebibliography}
\end{document}